\documentclass[11pt]{article}
\usepackage{amsmath, amsthm, amsfonts, amssymb, fullpage, indentfirst, graphicx, subfigure, cite}
\theoremstyle{plain}
\newtheorem{proposition}{Proposition}
\newtheorem{lemma}{Lemma}
\newtheorem{theorem}{Theorem}

\theoremstyle{definition}
\newtheorem{definition}{Definition}

\newtheorem{assumption}{Assumption}
\newtheorem{algorithm}{Algorithm}
\newtheorem{scenario}{Scenario}
\theoremstyle{remark}
\newtheorem{remark}{Remark}

\newenvironment{algorithmstep}{\ \begin{list}{\labelenumi}{\topsep0in\itemsep0in\parsep0in\labelwidth1in\usecounter{enumi}}}{\hfill$\blacksquare$\end{list}}

\begin{document}

\title{\LARGE\bf Distributed Estimation of Graph Spectrum\footnote{This work was supported by the National Science Foundation under grant CMMI-0900806.}}
\author{Mu Yang and Choon Yik Tang\\ School of Electrical and Computer Engineering\\ University of Oklahoma, Norman, OK 73019, USA\\ {\sf\{muyangwz,cytang\}@ou.edu}}
\date{\today}
\maketitle

\begin{abstract}
In this paper, we develop a two-stage distributed algorithm that enables nodes in a graph to cooperatively estimate the spectrum of a matrix $W$ associated with the graph, which includes the adjacency and Laplacian matrices as special cases. In the first stage, the algorithm uses a discrete-time linear iteration and the Cayley-Hamilton theorem to convert the problem into one of solving a set of linear equations, where each equation is known to a node. In the second stage, if the nodes happen to know that $W$ is cyclic, the algorithm uses a Lyapunov approach to asymptotically solve the equations with an exponential rate of convergence. If they do not know whether $W$ is cyclic, the algorithm uses a random perturbation approach and a structural controllability result to approximately solve the equations with an error that can be made small. Finally, we provide simulation results that illustrate the algorithm.
\end{abstract}

\section{Introduction}\label{sec:intr}

The spectrum of a graph, defined as the set of eigenvalues of either its adjacency or Laplacian matrix, provides a useful characterization of the properties of the graph. For instance, the distribution of such eigenvalues offers insights into the shapes and sizes of communities in a complex network \cite{Newman10}. As another example, the largest and smallest of such eigenvalues provides bounds on the maximum, minimum, and average node degrees \cite{Chung97}. The spectrum has also been used, for example, in chemistry, where it is associated with the stability of molecules \cite{Chung97}, and in quantum mechanics, where it is related to the energy of Hamiltonian systems \cite{Chung97}.

With the continued advances in technology that enable humans to build increasingly complex networks, it is becoming desirable that nodes in a network have the ability to analyze the network themselves, such as decentralizedly computing the spectrum of the network, so that valuable understanding about, say, the network structure may be gained. Motivated by this, a number of distributed algorithms have been proposed in the literature, including \cite{Sahai12, Franceschelli13, TranTMD14} that consider estimation of the entire spectrum of the Laplacian matrix, and \cite{YangP10, Aragues12, LiC13} that focus on estimation of its second smallest eigenvalue (i.e., the algebraic connectivity).

In this paper, we add to the literature by developing a two-stage distributed algorithm, which enables nodes in a graph to cooperatively estimate the spectrum of a matrix $W$ associated with the graph. Unlike in \cite{Sahai12, Franceschelli13, TranTMD14, YangP10, Aragues12, LiC13}, the matrix $W$ can be the adjacency or Laplacian matrix of the graph, a weighted version of these matrices, or any other matrix induced by the graph (see Section~\ref{sec:probform}). To construct the algorithm, we first use a discrete-time linear iteration and the Cayley-Hamilton theorem to convert the original problem into an equivalent problem of solving a set of linear equations of the form $Ax=b$, where every row of $A$ and $b$ is known to a particular node (Section~\ref{sec:formlineequa}). We then show that the matrix $A$ can be made almost surely nonsingular if the nodes happen to know that $W$ is cyclic, but not necessarily so if they do not (Section~\ref{sec:formlineequa}). In the case of the former, we use a Lyapunov approach to asymptotically solve the equations with an exponential rate of convergence (Section~\ref{sec:solvlineequascenknow}). In the case of the latter, we use a random perturbation approach and a structural controllability result to approximately solve the equations with an error that can be made small (Section~\ref{sec:solvlineequascennotknow}). Finally, we provide simulation results that illustrate our distributed algorithm (Section~\ref{sec:simuresu}) and conclude the paper with a word on future research directions (Section~\ref{sec:conc}).

\section{Problem Formulation}\label{sec:probform}

Consider a network modeled as an undirected, connected graph $\mathcal{G}=(\mathcal{V},\mathcal{E})$, where $\mathcal{V}=\{1,2,\ldots,N\}$ denotes the set of $N\ge2$ nodes and $\mathcal{E}\subset\{\{i,j\}:i,j\in\mathcal{V},i\ne j\}$ denotes the set of edges. Any two nodes $i,j\in\mathcal{V}$ are neighbors and can communicate if and only if $\{i,j\}\in\mathcal{E}$. The set of neighbors of each node $i\in\mathcal{V}$ is denoted as $\mathcal{N}_i=\{j\in\mathcal{V}:\{i,j\}\in\mathcal{E}\}$, and the communications are assumed to be delay- and error-free, with no quantization.

Suppose associated with the graph $\mathcal{G}$ is a square matrix $W=[w_{ij}]\in\mathbb{R}^{N\times N}$ satisfying the following assumption:

\begin{assumption}\label{asm:W}
The matrix $W$ is such that for each $i,j\in\mathcal{V}$ with $i\ne j$, if $\{i,j\}\notin\mathcal{E}$, then $w_{ij}=w_{ji}=0$.
\end{assumption}

Note that Assumption~\ref{asm:W} allows $w_{ii}$ $\forall i\in\mathcal{V}$ to be arbitrary. It also allows $w_{ij}$ and $w_{ji}$ $\forall\{i,j\}\in\mathcal{E}$ to be arbitrary and different. Thus, $W$ can be the adjacency or Laplacian matrix of graph $\mathcal{G}$, a weighted version of these matrices, or any other matrix associated with $\mathcal{G}$ as long as Assumption~\ref{asm:W} holds.

Suppose each node $i\in\mathcal{V}$ knows only $\mathcal{N}_i$, $w_{ii}$, and $w_{ij}$ $\forall j\in\mathcal{N}_i$, which it prefers to not share with any of its neighbors due perhaps to security and privacy reasons. Yet, despite having only such local information about the graph $\mathcal{G}$ and matrix $W$, suppose every node $i\in\mathcal{V}$ wants to determine the spectrum of $W$, i.e., all the $N$ eigenvalues of $W$, denoted as
\begin{align}
\lambda^{(1)},\lambda^{(2)},\ldots,\lambda^{(N)}\in\mathbb{C},\label{eq:lambda}
\end{align}
where complex eigenvalues must be in the form of conjugate pairs. Finally, suppose each node $i\in\mathcal{V}$ knows the value of $N$, which is not an unreasonable assumption since each of them wants to determine the values of $N$ objects.

Given the above, the goal of this paper is to devise a distributed algorithm that enables every node $i\in\mathcal{V}$ to estimate the spectrum \eqref{eq:lambda} of $W$ with a guaranteed accuracy.

\section{Forming a Set of Linear Equations}\label{sec:formlineequa}

In this section, we show that by having the nodes execute a discrete-time linear iteration $N$ times, the problem of finding the spectrum \eqref{eq:lambda} of $W$ may be converted into one of solving a set of linear equations with appealing properties.

Observe that although none of the nodes has complete information about $\mathcal{G}$ and $W$, each node $i\in\mathcal{V}$ knows the entire row $i$ of $W$ (since it knows $w_{ii}$ and $w_{ij}$ $\forall j\in\mathcal{N}_i$, and since $w_{ij}=0$ $\forall j\notin\{i\}\cup\mathcal{N}_i$ by Assumption~\ref{asm:W}). This makes the nodes well-suited to carry out the discrete-time linear iteration
\begin{align}
y_i(t+1)=w_{ii}y_i(t)+\sum_{j\in\mathcal{N}_i}w_{ij}y_j(t),\quad\forall i\in\mathcal{V},\;\forall t\in\mathbb{Z}_+,\label{eq:s1}
\end{align}
which in matrix form may be written as
\begin{align}
y(t+1)=Wy(t),\quad\forall t\in\mathbb{Z}_+,\label{eq:s2}
\end{align}
where $\mathbb{Z}_+=\{0,1,2,\ldots\}$, $y_i(t)\in\mathbb{R}$ is maintained in node $i$'s local memory, and
\begin{align}
y(t)=\begin{bmatrix}y_1(t) & y_2(t) & \cdots & y_N(t)\end{bmatrix}^T\in\mathbb{R}^N.\label{eq:s3}
\end{align}
Indeed, \eqref{eq:s1} or \eqref{eq:s2} can be implemented by having each node $i\in\mathcal{V}$ repeatedly send its $y_i(t)$ to every neighbor $j\in\mathcal{N}_i$.

Since \eqref{eq:s2} is a discrete-time linear system, we can write
\begin{align}
y(t)=W^ty(0),\quad\forall t\in\mathbb{Z}_+,\label{eq:s4}
\end{align}
so that
\begin{align}
y(N)=W^Ny(0).\label{eq:s5}
\end{align}
By the Cayley-Hamilton theorem, $W^N$ in \eqref{eq:s5} may be expressed as
\begin{align}
W^N=-x^{(0)}I_N-x^{(1)}W-\cdots-x^{(N-1)}W^{N-1},\label{eq:s6}
\end{align}
where $I_n\in\mathbb{R}^{n\times n}$ is the identity matrix and the scalars $x^{(0)},x^{(1)},\ldots,x^{(N-1)}\in\mathbb{R}$ are the $N$ coefficients of the characteristic polynomial of $W$, i.e.,
\begin{align}
&\operatorname{det}(\lambda I_N-W)=(\lambda-\lambda^{(1)})(\lambda-\lambda^{(2)})\cdots(\lambda-\lambda^{(N)})\nonumber\\
&\quad=\lambda^N+x^{(N-1)}\lambda^{N-1}+\cdots+x^{(1)}\lambda+x^{(0)}.\label{eq:s7}
\end{align}
Substituting \eqref{eq:s6} into \eqref{eq:s5} and using \eqref{eq:s4}, we obtain
\begin{align}
y(N)&=(-x^{(0)}I_N-x^{(1)}W-\cdots-x^{(N-1)}W^{N-1})y(0)\nonumber\\
&=-x^{(0)}y(0)-x^{(1)}y(1)-\cdots-x^{(N-1)}y(N-1).\label{eq:s8}
\end{align}
By using \eqref{eq:s3}, we can rewrite \eqref{eq:s8} as
\begin{align}
\underbrace{\begin{bmatrix}y_1(0) & y_1(1) & \cdots & y_1(N-1)\\ y_2(0) & y_2(1) & \cdots & y_2(N-1)\\ \vdots & \vdots & \ddots & \vdots\\ y_N(0) & y_N(1) & \cdots & y_N(N-1)\end{bmatrix}}_{A}\underbrace{\begin{bmatrix}x^{(0)}\\ x^{(1)}\\ \vdots\\ x^{(N-1)}\end{bmatrix}}_{x^*}=\underbrace{\begin{bmatrix}-y_1(N)\\ -y_2(N)\\ \vdots\\ -y_N(N)\end{bmatrix}}_{b},\label{eq:s9}
\end{align}
where, for later convenience, we denote the matrix on the left-hand side of \eqref{eq:s9} as $A\in\mathbb{R}^{N\times N}$, the vector of characteristic polynomial coefficients as $x^*\in\mathbb{R}^N$, and the vector on the right-hand side of \eqref{eq:s9} as $b\in\mathbb{R}^N$.

The matrix equation \eqref{eq:s9} suggests the following approach for finding the spectrum \eqref{eq:lambda} of $W$: suppose each node $i\in\mathcal{V}$ selects an initial condition $y_i(0)\in\mathbb{R}$. Upon selecting the $y_i(0)$'s, suppose the nodes execute the discrete-time linear iteration \eqref{eq:s1} or equivalently \eqref{eq:s2} $N$ times for $t\in\{0,1,\ldots,N-1\}$. During the execution, suppose each node $i\in\mathcal{V}$ stores the resulting $N+1$ numbers $y_i(0),y_i(1),\ldots,y_i(N-1),y_i(N)$ in its local memory. Then, \eqref{eq:s9} is a set of $N$ linear equations in which each node $i\in\mathcal{V}$ knows the entire row $i$ of $A$ and $b$, and in which the vector $x^*$ of $N$ characteristic polynomial coefficients $x^{(0)},x^{(1)},\ldots,x^{(N-1)}$ of $W$ are the $N$ unknowns. It follows that if $A$ is nonsingular, and if the nodes are able to cooperatively solve \eqref{eq:s9} for the unique $x^*$, then each of them could determine on its own the $N$ eigenvalues $\lambda^{(1)},\lambda^{(2)},\ldots,\lambda^{(N)}$ of $W$ using \eqref{eq:s7} and a polynomial root-finding algorithm.

To realize the above approach, it is necessary that $A$ in \eqref{eq:s9} is nonsingular. To see whether this can be ensured, observe from \eqref{eq:s3}, \eqref{eq:s4}, and \eqref{eq:s9} that $A$ may be expressed as
\begin{align}
A=\begin{bmatrix}y(0) & Wy(0) & \cdots & W^{N-1}y(0)\end{bmatrix}.\label{eq:s10}
\end{align}
In the form \eqref{eq:s10}, $A$ is, interestingly, the controllability matrix of a fictitious discrete-time single-input linear system
\begin{align}
z(t+1)=Wz(t)+y(0)u(t),\quad\forall t\in\mathbb{Z}_+,\label{eq:s11}
\end{align}
where $z(t)\in\mathbb{R}^N$ is its state, $u(t)\in\mathbb{R}$ is its input, $W$ is its state matrix, and $y(0)$ is its input matrix. Hence:

\begin{proposition}\label{pro:AWy0}
The matrix $A$ in \eqref{eq:s9} or \eqref{eq:s10} is nonsingular if and only if the pair $(W,y(0))$ of the system \eqref{eq:s11} is controllable.
\end{proposition}

Since $W$ is given by the problem but $y(0)$ may be freely selected by the nodes, it may be possible to select $y(0)$ so that the pair $(W,y(0))$ is controllable. The following definition and lemmas examine this possibility:

\begin{definition}[\!\!\cite{ZhouK96}]\label{def:cyc}
A square matrix with real entries is said to be {\em cyclic} if each of its distinct eigenvalues has a geometric multiplicity of $1$.
\end{definition}

\begin{lemma}\label{lem:Wnotcyc}
If $W$ is not cyclic, then for every $y(0)\in\mathbb{R}^N$, the pair $(W,y(0))$ is not controllable.
\end{lemma}

\begin{proof}
Suppose $W$ is not cyclic and let $y(0)\in\mathbb{R}^N$ be given. Then, by Definition~\ref{def:cyc}, $W$ has an eigenvalue $\lambda\in\mathbb{C}$ whose geometric multiplicity exceeds $1$, i.e., $\operatorname{rank}(W-\lambda I_N)<N-1$. Since $y(0)$ is a column vector, $\operatorname{rank}([W-\lambda I_N\;|\;y(0)])<N$. Therefore, by statements~(i) and~(iv) of Theorem~3.1 in \cite{ZhouK96}, the pair $(W,y(0))$ is not controllable.
\end{proof}

\begin{lemma}\label{lem:Wcyc}
If $W$ is cyclic, then for almost every $y(0)\in\mathbb{R}^N$, the pair $(W,y(0))$ is controllable.
\end{lemma}

\begin{proof}
According to Lemma~3.12 in \cite{ZhouK96}, if $\mathcal{A}\in\mathbb{R}^{n\times n}$ is cyclic and $\mathcal{B}\in\mathbb{R}^{n\times m}$ is such that the pair $(\mathcal{A},\mathcal{B})$ is controllable, then for almost every $v\in\mathbb{R}^m$, the pair $(\mathcal{A},\mathcal{B}v)$ is controllable. Applying this lemma with $\mathcal{A}=W$, $\mathcal{B}=I_N$, and $v=y(0)$, and using the fact that the pair $(W,I_N)$ is controllable, we conclude that so is the pair $(W,y(0))$.
\end{proof}

Proposition~\ref{pro:AWy0} and Lemma~\ref{lem:Wnotcyc} imply that $W$ being cyclic is necessary for $A$ in \eqref{eq:s9} or \eqref{eq:s10} to be nonsingular. Lemma~\ref{lem:Wcyc}, on the other hand, implies that $W$ being cyclic is essentially sufficient because almost every $y(0)\in\mathbb{R}^N$ would work. This latter result is especially useful in a decentralized network because the result allows each node $i\in\mathcal{V}$ to select its $y_i(0)\in\mathbb{R}$ independently from other nodes and randomly from any continuous probability distribution before executing \eqref{eq:s1} or \eqref{eq:s2}, and be almost sure that the resulting $A$ would be nonsingular.

Motivated by the above analysis, in the rest of this paper we consider separately the following two scenarios:

\begin{scenario}\label{sce:know}
The nodes know that $W$ is cyclic.
\end{scenario}

\begin{scenario}\label{sce:notknow}
The nodes do not know whether $W$ is cyclic, or know that $W$ is not cyclic.
\end{scenario}

We consider Scenarios~\ref{sce:know} and~\ref{sce:notknow} separately because Scenario~\ref{sce:know} is easier to deal with (in Section~\ref{sec:solvlineequascenknow}) and its treatment helps us deal with Scenario~\ref{sce:notknow} (in Section~\ref{sec:solvlineequascennotknow}). We note that both of these scenarios arise in applications. For instance, if the graph $\mathcal{G}$ represents a sensor network and the entries $w_{ii}$ $\forall i\in\mathcal{V}$ and $w_{ij}$ $\forall\{i,j\}\in\mathcal{E}$ of $W$ represent random sensor measurements with continuous probability distributions, then Scenario~\ref{sce:know} takes place as the nodes could say with near certainty that $W$ is cyclic because almost every $n$-by-$n$ matrix has $n$ distinct eigenvalues and, thus, is cyclic. In contrast, if $W$ represents the adjacency or Laplacian matrix of $\mathcal{G}$, then Scenario~\ref{sce:notknow} takes place as $W$ would be cyclic if $\mathcal{G}$ is, say, a path graph \cite{Chung97} and would not be cyclic if $\mathcal{G}$ is, say, a complete or cycle graph \cite{Chung97}, which the nodes could not tell because they only have local information about $\mathcal{G}$.

To summarize, in this section we have transformed the problem of finding the spectrum \eqref{eq:lambda} of $W$ into one of solving the set of linear equations \eqref{eq:s9}, in which each node $i\in\mathcal{V}$ knows the entire row $i$ of $A$ and $b$, and in which $A$ can be made almost surely nonsingular in Scenario~\ref{sce:know}, but not necessarily so in Scenario~\ref{sce:notknow}.

\section{Solving the Set of Linear Equations in Scenario~\ref{sce:know}}\label{sec:solvlineequascenknow}

In this section, we focus on Scenario~\ref{sce:know} and develop a continuous-time distributed algorithm that enables the nodes to asymptotically solve the set of linear equations \eqref{eq:s9} with an exponential rate of convergence.

To facilitate the development, we assume that the nodes have executed \eqref{eq:s1} or \eqref{eq:s2} to arrive at \eqref{eq:s9}. Moreover, since $A$ in \eqref{eq:s9} can be made almost surely nonsingular in this Scenario~\ref{sce:know}, we assume that it {\em is} nonsingular throughout the section. With these assumptions, for each $i\in\mathcal{V}$ let $a_i=\begin{bmatrix}y_i(0) & y_i(1) & \cdots & y_i(N-1)\end{bmatrix}^T\in\mathbb{R}^N$ and $b_i=-y_i(N)\in\mathbb{R}$, so that \eqref{eq:s9} may be stated as
\begin{align}
\underbrace{\begin{bmatrix}\;\text{---}\;a_1^T\;\text{---}\;\\ \;\text{---}\;a_2^T\;\text{---}\;\\ \vdots\\ \;\text{---}\;a_N^T\;\text{---}\;\end{bmatrix}}_{A}\underbrace{\begin{bmatrix}x^{(0)}\\ x^{(1)}\\ \vdots\\ x^{(N-1)}\end{bmatrix}}_{x^*}=\underbrace{\begin{bmatrix}b_1\\ b_2\\ \vdots\\ b_N\end{bmatrix}}_{b},\label{eq:Ax*=b}
\end{align}
where $a_i$ and $b_i$ are known to node $i$ because \eqref{eq:s1} or \eqref{eq:s2} has been executed. In addition to knowing $a_i$ and $b_i$, suppose each node $i\in\mathcal{V}$ maintains in its local memory an estimate $x_i(t)=\begin{bmatrix}x_i^{(0)}(t) & x_i^{(1)}(t) & \cdots & x_i^{(N-1)}(t)\end{bmatrix}^T\in\mathbb{R}^N$ of the unknown, unique solution $x^*\in\mathbb{R}^N$, where here $t\in[0,\infty)$ denotes continuous-time (unlike in Section~\ref{sec:formlineequa} where $t\in\mathbb{Z}_+$ denotes discrete-time). Furthermore, let $\mathbf{x}(t)=(x_1(t),x_2(t),\ldots,x_N(t))\in\mathbb{R}^{N^2}$ and $\mathbf{x}^*=(x^*,x^*,\ldots,x^*)\in\mathbb{R}^{N^2}$ be vectors obtained by stacking the $N$ estimates $x_i(t)$'s and $N$ copies of the solution $x^*$.

To come up with a distributed algorithm that gradually drives $\mathbf{x}(t)$ to $\mathbf{x}^*$, consider a quadratic Lyapunov function candidate $V:\mathbb{R}^{N^2}\rightarrow\mathbb{R}$, defined as
\begin{align}
V(\mathbf{x})=\sum_{i\in\mathcal{V}}\alpha_i(a_i^Tx_i-b_i)^2+\sum_{\{i,j\}\in\mathcal{E}}\beta_{\{i,j\}}(x_i-x_j)^T(x_i-x_j),\label{eq:V}
\end{align}
where $\alpha_i>0$ $\forall i\in\mathcal{V}$ and $\beta_{\{i,j\}}>0$ $\forall\{i,j\}\in\mathcal{E}$ are parameters. Notice that each term in the first summation in \eqref{eq:V} is a measure of how far away from the hyperplane $\{z\in\mathbb{R}^N:a_i^Tz=b_i\}$ the estimate $x_i(t)$ is. Moreover, because $A$ is nonsingular and because of \eqref{eq:Ax*=b}, the $N$ hyperplanes $\{z\in\mathbb{R}^N:a_i^Tz=b_i\}$ $\forall i\in\mathcal{V}$ have a unique intersection at $x^*$. Furthermore, the second summation in \eqref{eq:V} is a measure of the disagreement among the estimates $x_i(t)$'s. Hence, both the first and second summations in \eqref{eq:V} are only positive semidefinite functions of $\mathbf{x}$. However, as the following proposition shows, adding them up makes $V$ a legitimate Lyapunov function candidate:

\begin{proposition}\label{pro:AVpd}
If $A$ in \eqref{eq:Ax*=b} is nonsingular, then the function $V$ in \eqref{eq:V} is positive definite with respect to $\mathbf{x}^*$.
\end{proposition}

\begin{proof}
Clearly, $V$ is a positive semidefinite function of $\mathbf{x}$. To show that it is positive definite with respect to $\mathbf{x}^*$, we show that $V(\mathbf{x})=0$ if and only if $\mathbf{x}=\mathbf{x}^*$. Suppose $\mathbf{x}=\mathbf{x}^*$. Then, $a_i^Tx_i-b_i=0$ $\forall i\in\mathcal{V}$ according to \eqref{eq:Ax*=b}. In addition, the second summation in \eqref{eq:V} drops out. Therefore, $V(\mathbf{x})=0$. Next, suppose $V(\mathbf{x})=0$. Then,
\begin{align}
a_i^Tx_i&=b_i,\quad\forall i\in\mathcal{V},\label{eq:aixi=bi}\\
x_i&=x_j,\quad\forall\{i,j\}\in\mathcal{E}.\label{eq:xi=xj}
\end{align}
Since $\mathcal{G}$ is connected, \eqref{eq:xi=xj} implies that there exists $\tilde{x}\in\mathbb{R}^N$ such that $x_i=\tilde{x}$ $\forall i\in\mathcal{V}$. Substituting this into \eqref{eq:aixi=bi}, we get $a_i^T\tilde{x}=b_i$ $\forall i\in\mathcal{V}$ or, equivalently, $A\tilde{x}=b$. Since $A$ is nonsingular, we have $\tilde{x}=x^*$, so that $\mathbf{x}=\mathbf{x}^*$.
\end{proof}

\begin{remark}\label{rem:V}
Notice that $V$ in \eqref{eq:V} can also be written as
\begin{align*}
V(\mathbf{x})=(\mathbf{x}-\mathbf{x}^*)^TP(\mathbf{x}-\mathbf{x}^*),
\end{align*}
where $P=P^T\in\mathbb{R}^{N^2\times N^2}$ is positive definite and given by
\begin{align*}
P=\begin{bmatrix}\alpha_1a_1a_1^T & & & 0\\ & \alpha_2a_2a_2^T & &\\ & & \ddots &\\ 0 & & &\alpha_Na_Na_N^T\end{bmatrix}+L_\beta\otimes I_N,
\end{align*}
where $\otimes$ denotes the Kronecker product and $L_\beta=[L_{ij}]\in\mathbb{R}^{N\times N}$ is a weighted Laplacian matrix of $\mathcal{G}$ with $L_{ii}=\sum_{j\in\mathcal{N}_i}\beta_{\{i,j\}}$, $L_{ij}=-\beta_{\{i,j\}}$ if $\{i,j\}\in\mathcal{E}$, and $L_{ij}=0$ if $i\ne j$ and $\{i,j\}\notin\mathcal{E}$.\qed
\end{remark}

With Proposition~\ref{pro:AVpd} in hand, we next take the time derivative of $V$ along the state trajectory $\mathbf{x}(t)$ to obtain
\begin{align}
\dot{V}(\mathbf{x}(t))=2\sum_{i\in\mathcal{V}}\Bigl[\alpha_i(a_i^Tx_i(t)-b_i)a_i+\sum_{j\in\mathcal{N}_i}\beta_{\{i,j\}}(x_i(t)-x_j(t))\Bigr]\dot{x}_i(t),\quad\forall t\in[0,\infty).\label{eq:Vdot}
\end{align}
Examining \eqref{eq:Vdot}, we see that $\dot{V}(\mathbf{x}(t))$ can be made negative semidefinite---at the very least---by letting each $\dot{x}_i(t)$ be the negative of the expression within the brackets in \eqref{eq:Vdot}, i.e.,
\begin{align}
\dot{x}_i(t)=-\alpha_i(a_i^Tx_i(t)-b_i)a_i-\sum_{j\in\mathcal{N}_i}\beta_{\{i,j\}}(x_i(t)-x_j(t)),\quad\forall i\in\mathcal{V},\;\forall t\in[0,\infty).\label{eq:xidot}
\end{align}
The following theorem asserts that the continuous-time system \eqref{eq:xidot} possesses an excellent property:

\begin{theorem}\label{thm:uniqeqptges}
If $A$ in \eqref{eq:Ax*=b} is nonsingular, then the system \eqref{eq:xidot} has a unique equilibrium point at $\mathbf{x}^*$ that is globally exponentially stable, so that $\forall\mathbf{x}(0)\in\mathbb{R}^{N^2}$, $\lim_{t\rightarrow\infty}\mathbf{x}(t)=\mathbf{x}^*$, i.e., $\lim_{t\rightarrow\infty}x_i(t)=x^*$ $\forall i\in\mathcal{V}$.
\end{theorem}

\begin{proof}
For each $i\in\mathcal{V}$, setting $\dot{x}_i(t)$ in \eqref{eq:xidot} to zero yields
\begin{align}
0=-\alpha_i(a_i^Tx_i-b_i)a_i-\sum_{j\in\mathcal{N}_i}\beta_{\{i,j\}}(x_i-x_j).\label{eq:t1}
\end{align}
Summing both sides of \eqref{eq:t1} over $i\in\mathcal{V}$ gives
\begin{align}
0=\sum_{i\in\mathcal{V}}-\alpha_i(a_i^Tx_i-b_i)a_i.\label{eq:t2}
\end{align}
Due to \eqref{eq:Ax*=b} and to $A$ being nonsingular, the vectors $a_1,a_2,\ldots,a_N$ in \eqref{eq:t2} are linearly independent in $\mathbb{R}^N$. Thus,
\begin{align}
0=-\alpha_i(a_i^Tx_i-b_i),\quad\forall i\in\mathcal{V}.\label{eq:t3}
\end{align}
Substituting \eqref{eq:t3} back into \eqref{eq:t1} results in
\begin{align*}
0=\sum_{j\in\mathcal{N}_i}\beta_{\{i,j\}}(x_i-x_j),\quad\forall i\in\mathcal{V},
\end{align*}
which is equivalent to
\begin{align}
0=(L_\beta\otimes I_N)\mathbf{x},\label{eq:t4}
\end{align}
where $\otimes$ and $L_\beta$ have been defined in Remark~\ref{rem:V}. Since $\mathcal{G}$ is connected, \eqref{eq:t4} implies that $x_i=\tilde{x}$ $\forall i\in\mathcal{V}$ for some $\tilde{x}\in\mathbb{R}^N$. Plugging this into \eqref{eq:t3} yields $a_i^T\tilde{x}=b_i$ $\forall i\in\mathcal{V}$. Since $A$ is nonsingular, we have $\tilde{x}=x^*$, i.e., $\mathbf{x}=\mathbf{x}^*$. Hence, the system \eqref{eq:xidot} has a unique equilibrium point at $\mathbf{x}^*$. Since for each $i\in\mathcal{V}$ the right-hand side of \eqref{eq:xidot} is the negative of the expression within the brackets in \eqref{eq:Vdot}, $\dot{V}(\mathbf{x}(t))$ is negative definite with respect to $\mathbf{x}^*$. Therefore, the equilibrium point $\mathbf{x}^*$ is globally exponentially stable.
\end{proof}

Having established Theorem~\ref{thm:uniqeqptges}, we now relate it back to the original problem of finding the spectrum \eqref{eq:lambda} of $W$. To this end, suppose each node $i\in\mathcal{V}$ maintains in its local memory an estimate $\lambda_i^{(\ell)}(t)\in\mathbb{C}$ of the unknown, $\ell$th eigenvalue $\lambda^{(\ell)}$ of $W$ for $\ell\in\{1,2,\ldots,N\}$. Also suppose at each time $t\in[0,\infty)$, node $i$ lets its $N$ estimates $\lambda_i^{(\ell)}(t)$'s be the roots of an $N$th-order polynomial formed by the estimate $x_i(t)=\begin{bmatrix}x_i^{(0)}(t) & x_i^{(1)}(t) & \cdots & x_i^{(N-1)}(t)\end{bmatrix}^T$ that is also stored in its local memory, i.e.,
\begin{align}
&(\lambda-\lambda_i^{(1)}(t))(\lambda-\lambda_i^{(2)}(t))\cdots(\lambda-\lambda_i^{(N)}(t))\nonumber\\
&\quad=\lambda^N+x_i^{(N-1)}(t)\lambda^{N-1}+\cdots+x_i^{(1)}(t)\lambda+x_i^{(0)}(t),\quad\forall i\in\mathcal{V},\;\forall t\in[0,\infty),\label{eq:rootfind}
\end{align}
which can be implemented using a polynomial root-finding algorithm embedded in node $i$. Then, because $(\lambda^{(1)},\lambda^{(2)},\ldots,\lambda^{(N)})$ in \eqref{eq:s7} is a continuous function of $x^*$, and $(\lambda_i^{(1)}(t),\lambda_i^{(2)}(t),\ldots,\lambda_i^{(N)}(t))$ in \eqref{eq:rootfind} is the {\em same} continuous function of $x_i(t)$, Theorem~\ref{thm:uniqeqptges} implies that
\begin{align}
\lim_{t\rightarrow\infty}\lambda_i^{(\ell)}(t)=\lambda^{(\ell)},\quad\forall i\in\mathcal{V},\;\forall\ell\in\{1,2,\ldots,N\}.\label{eq:limlambda}
\end{align}
Equation \eqref{eq:limlambda}, in turn, implies that the system \eqref{eq:xidot} is a continuous-time distributed algorithm that enables the nodes to asymptotically learn the spectrum \eqref{eq:lambda} of $W$.

Putting together the development in Sections~\ref{sec:formlineequa} and~\ref{sec:solvlineequascenknow}, we obtain the following two-stage distributed algorithm, which is applicable to this Scenario~\ref{sce:know}:

\begin{algorithm}[For Scenario~\ref{sce:know}]\label{alg:know}
\begin{algorithmstep}
\item Each node $i\in\mathcal{V}$ selects its $y_i(0)\in\mathbb{R}$ independently from other nodes and randomly from any continuous probability distribution.
\item Upon completion, the nodes execute \eqref{eq:s1} or \eqref{eq:s2} $N$ times for $t\in\{0,1,\ldots,N-1\}$, so that each node $i\in\mathcal{V}$ gradually learns the entire row $i$ of $A$ and $b$ in \eqref{eq:s9}.
\item Upon completion, the nodes execute \eqref{eq:xidot} and \eqref{eq:rootfind} for $t\in[0,\infty)$, so that each node $i\in\mathcal{V}$ is able to continuously update its $x_i(t)$ and $\lambda_i^{(\ell)}(t)$'s.
\end{algorithmstep}
\end{algorithm}

\begin{remark}\label{rem:lite}
The current literature offers a few distributed algorithms \cite{Nedic10, MouS13} that may be used to solve linear equations \eqref{eq:s9}. These algorithms are different from \eqref{eq:xidot} in that they force the state of each node to stay in an affine set, whereas \eqref{eq:xidot} allows the state to freely roam the state space.\qed
\end{remark}

\section{Solving the Set of Linear Equations in Scenario~\ref{sce:notknow}}\label{sec:solvlineequascennotknow}

In this section, we focus on Scenario~\ref{sce:notknow} and provide a slightly different algorithm that enables the nodes to approximately solve \eqref{eq:s9} with an error that can be made small.

Recall that Scenario~\ref{sce:notknow} represents a situation where the nodes either do not know whether $W$ is cyclic, or somehow know that $W$ is not cyclic. Consequently, they either do not know whether $A$ in \eqref{eq:s9} is nonsingular, or know that $A$ is singular. Although the nodes could still apply Algorithm~\ref{alg:know}, there is no guarantee that their estimates $x_i(t)$'s would converge to $x^*$. One way to address this issue is to have the nodes randomly perturb the matrix $W$ and vector $y(0)$, so that the resulting $A$ in \eqref{eq:s10} is hopefully nonsingular. Of course, such a random perturbation approach no longer allows them to asymptotically determine the exact spectrum of $W$. However, getting an estimate of the spectrum of $W$ may be sufficient in some applications. Thus, we will adopt this random perturbation approach in this Scenario~\ref{sce:notknow}.

For notational simplicity, let the matrix associated with the graph $\mathcal{G}$ be denoted as $\overline{W}=[\overline{w}_{ij}]\in\mathbb{R}^{N\times N}$ instead of $W=[w_{ij}]$, and let $W$ instead denote a perturbed version of $\overline{W}$. In addition, let $\overline{x}^{(\ell)}$'s and $\overline{\lambda}^{(\ell)}$'s denote, respectively, the characteristic polynomial coefficients and eigenvalues of $\overline{W}$ that the nodes wish to determine, and let $x^{(\ell)}$'s and $\lambda^{(\ell)}$'s denote those of $W$ as before. Moreover, let the perturbed matrix $W$ be obtained from $\overline{W}$ in a decentralized manner as follows: prior to executing \eqref{eq:s1} or \eqref{eq:s2}, each node $i\in\mathcal{V}$ lets
\begin{align}
w_{ii}&=\overline{w}_{ii}+\delta_{ii},\quad\forall i\in\mathcal{V},\label{eq:wii}\\
w_{ij}&=\overline{w}_{ij}+\delta_{ij},\quad\forall i\in\mathcal{V},\;\forall j\in\mathcal{N}_i,\label{eq:wij}
\end{align}
where the $\delta_{ii}$'s and $\delta_{ij}$'s are independent, uniformly distributed random variables in the interval $[-a,a]$, so that $a>0$ represents the perturbation magnitude. Notice that since $\overline{w}_{ij}=0$ $\forall i\in\mathcal{V}$ $\forall j\notin\{i\}\cup\mathcal{N}_i$ by Assumption~\ref{asm:W},
\begin{align}
w_{ij}=0,\quad\forall i\in\mathcal{V},\;\forall j\notin\{i\}\cup\mathcal{N}_i\label{eq:wij0}
\end{align}
as well. Also note that because the nodes are slated to select their $y_i(0)$'s independently and randomly from a continuous probability distribution, there is no need to further randomly perturb these $y_i(0)$'s.

The following lemma uses a structural controllability result to show that the aforementioned approach is effective:

\begin{lemma}\label{lem:randpert}
If $W$ is as defined in \eqref{eq:wii}--\eqref{eq:wij0} and $y(0)$ is as defined in Step~1 of Algorithm~\ref{alg:know}, then $A$ in \eqref{eq:s10} is almost surely nonsingular.
\end{lemma}

\begin{proof}
Reconsider the graph $\mathcal{G}=(\mathcal{V},\mathcal{E})$ from Section~\ref{sec:probform}. Let $\mathcal{S}=\{(\mathcal{A},\mathcal{B})\in\mathbb{R}^{N\times N}\times\mathbb{R}^N:\mathcal{A}_{ij}=0\;\text{if}\;i\ne j\;\text{and}\;\{i,j\}\notin\mathcal{E}\}$ and $\mathcal{S}_c=\{(\mathcal{A},\mathcal{B})\in\mathcal{S}:(\mathcal{A},\mathcal{B})\;\text{is controllable}\}\subset\mathcal{S}$. In addition, let $\mathcal{A}^*=\operatorname{diag}(1,2,\ldots,N)\in\mathbb{R}^{N\times N}$ and $\mathcal{B}^*\in\mathbb{R}^N$ be the all-one vector. Then, $(\mathcal{A}^*,\mathcal{B}^*)\in\mathcal{S}$ according to the definition of $\mathcal{S}$. Moreover, $(\mathcal{A}^*,\mathcal{B}^*)\in\mathcal{S}_c$ because the controllability matrix formed by $(\mathcal{A}^*,\mathcal{B}^*)$ is a Vandermonde matrix that is nonsingular. These two properties of $(\mathcal{A}^*,\mathcal{B}^*)$, along with the definition of structural controllability \cite{LinCT74}, imply that every $(\mathcal{A},\mathcal{B})\in\mathcal{S}$ is structurally controllable. Next, let $(\mathcal{A},\mathcal{B})\in\mathcal{S}$ and $\epsilon>0$ be given. Then, by Proposition~1 of \cite{LinCT74}, there exists $(\mathcal{A}_c,\mathcal{B}_c)\in\mathcal{S}_c$ such that $\|\mathcal{A}-\mathcal{A}_c\|<\epsilon$ and $\|\mathcal{B}-\mathcal{B}_c\|<\epsilon$. Hence, $\mathcal{S}_c$ is a dense subset of $\mathcal{S}$. Lastly, note that $(W,y(0))\in\mathcal{S}$ due to Assumption~\ref{asm:W}, \eqref{eq:wii}--\eqref{eq:wij0}, and Step~1 of Algorithm~\ref{alg:know}. Since $\mathcal{S}_c$ is a dense subset of $\mathcal{S}$, $(W,y(0))$ is almost surely in $\mathcal{S}_c$. Therefore, by Proposition~\ref{pro:AWy0}, $A$ in \eqref{eq:s10} is almost surely nonsingular.
\end{proof}

As it follows from Lemma~\ref{lem:randpert}, by having the nodes perform the extra step described in \eqref{eq:wii}--\eqref{eq:wij0}, the results developed in Sections~\ref{sec:formlineequa} and~\ref{sec:solvlineequascenknow} become applicable to this Scenario~\ref{sce:notknow}. Furthermore, because both the characteristic polynomial coefficients and eigenvalues of a matrix are continuous functions of its entries, by having the nodes decrease the perturbation magnitude $a$ toward zero, the differences between the $x^{(\ell)}$'s and $\lambda^{(\ell)}$'s of $W$ and the $\overline{x}^{(\ell)}$'s and $\overline{\lambda}^{(\ell)}$'s of $\overline{W}$ can be made arbitrarily small, at least in principle. Note, however, that numerical issues may arise when $a$ is too small, or when the resulting $A$ is ill-conditioned. At present, we do not have answers to these numerical issues, and we believe they are important future research directions.

Based on the above, we obtain the following two-stage distributed algorithm for this Scenario~\ref{sce:notknow}:

\begin{algorithm}[For Scenario~\ref{sce:notknow}]\label{alg:notknow}
\begin{algorithmstep}
\item Each node $i\in\mathcal{V}$ executes \eqref{eq:wii}--\eqref{eq:wij0} to obtain a perturbed matrix $W$.
\item The remaining steps are identical to those of Algorithm~\ref{alg:know}.
\end{algorithmstep}
\end{algorithm}

\section{Simulation Results}\label{sec:simuresu}

In this section, we present two sets of simulation results that demonstrate the effectiveness of Algorithm~\ref{alg:know} for Scenario~\ref{sce:know} and Algorithm~\ref{alg:notknow} for Scenario~\ref{sce:notknow}.

\subsection{Simulation of Algorithm~\ref{alg:know} for Scenario~\ref{sce:know}}\label{ssec:simualgoscenknow}

\begin{figure}[tb]
\centering\subfigure[A $6$-node graph.]{\includegraphics[width=0.48\linewidth]{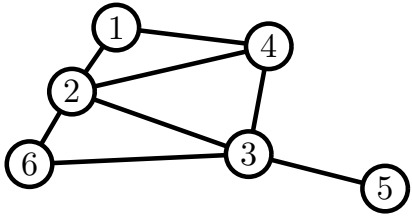}}\quad\subfigure[Data points $y_i(t)$ for $i\in\{1,2,\ldots,6\}$ and $t\in\{0,1,\ldots,6\}$ that form the set of linear equations \eqref{eq:s9}.]{\includegraphics[width=0.48\linewidth]{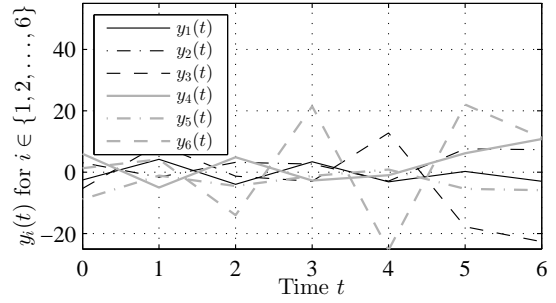}}\\ \subfigure[Node $3$'s estimate $x_3^{(\ell)}(t)$ of the $\ell$th characteristic polynomial coefficient $x^{(\ell)}$ for $\ell\in\{0,1,\ldots,5\}$.]{\includegraphics[width=0.48\linewidth]{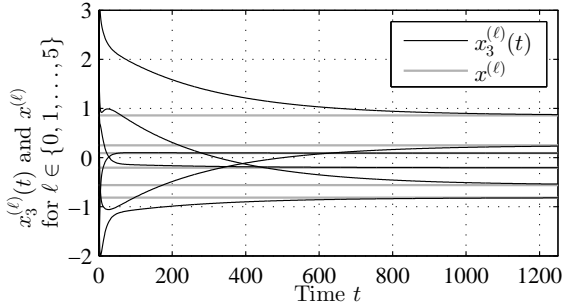}}\quad\subfigure[Node $i$'s estimate $x_i^{(1)}(t)$ of the $1$st characteristic polynomial coefficient $x^{(1)}$ for $i\in\{1,2,\ldots,6\}$.]{\includegraphics[width=0.48\linewidth]{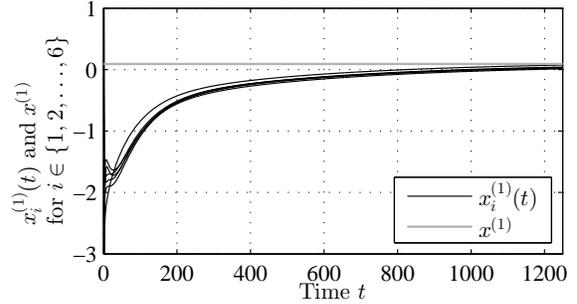}}
\caption{Performance of Algorithm~\ref{alg:know} for Scenario~\ref{sce:know}.}
\label{fig:sce1}
\end{figure}

Consider a sensor network with $N=6$ nodes, modeled as an undirected, connected graph $\mathcal{G}$, whose topology is shown in Figure~\ref{fig:sce1}(a). Suppose associated with the graph $\mathcal{G}$ is a $6$-by-$6$ matrix $W$, whose entries satisfy Assumption~\ref{asm:W} and represent random sensor measurements given by
\begin{align*}
W=\begin{bmatrix}-0.10 & -0.24 & 0 & 0.78 & 0 & 0\\ 0.24 & 0.53 & 0.39 & -0.04 & 0 & -0.19\\ 0 & 0.34 & 0.21 & 1.15 & -0.13 & 0.71\\-0.26 & -0.21 & 0.32 & -0.54 & 0 & 0\\ 0 & 0 & -0.45 & 0 & 0.39 & 0\\ 0 & 0.47 & -0.84 & 0 & 0 & -1.35\end{bmatrix}.
\end{align*}
Assuming that such measurements are realizations of continuously distributed random variables, the nodes are almost certain that $W$ is cyclic, so that Scenario~\ref{sce:know} takes place. Thus, to determine all the eigenvalues $\lambda^{(\ell)}$'s of $W$, which are given by $-1.02\pm0.55i$, $-0.004\pm0.46i$, $0.38$, and $0.81$, the nodes may apply Algorithm~\ref{alg:know}.

Figures~\ref{fig:sce1}(b)--\ref{fig:sce1}(d) display the result of simulating Algorithm~\ref{alg:know} with $\alpha_i=10$ $\forall i\in\mathcal{V}$ and $\beta_{\{i,j\}}=10$ $\forall\{i,j\}\in\mathcal{E}$. Specifically, Figure~\ref{fig:sce1}(b) shows the data points $y_i(t)$ for $i\in\{1,2,\ldots,6\}$ and $t\in\{0,1,\ldots,6\}$ that are used to form the set of linear equations \eqref{eq:s9}. Figure~\ref{fig:sce1}(c) shows, as a function of time $t$, node $3$'s estimate $x_3^{(\ell)}(t)$ of the $\ell$th characteristic polynomial coefficient $x^{(\ell)}$ of $W$ for $\ell\in\{0,1,\ldots,5\}$. Likewise, Figure~\ref{fig:sce1}(d) shows node $i$'s estimate $x_i^{(1)}(t)$ of the $1$st coefficient $x^{(1)}$ for $i\in\{1,2,\ldots,6\}$. (Due to space limitation, we are unable to include plots of $x_i^{(\ell)}(t)$ for all $i\in\{1,2,\ldots,6\}$ and $\ell\in\{0,1,\ldots,5\}$.) Observe that despite having only local information about $\mathcal{G}$ and $W$, the nodes are able to utilize Algorithm~\ref{alg:know} to asymptotically determine all the characteristic polynomial coefficients $x^{(\ell)}$'s of $W$ and, hence, all its eigenvalues $\lambda^{(\ell)}$'s.

\subsection{Simulation of Algorithm~\ref{alg:notknow} for Scenario~\ref{sce:notknow}}\label{ssec:simualgoscennotknow}

\begin{figure}[tb]
\centering\subfigure[A $6$-node graph.]{\includegraphics[width=0.48\linewidth]{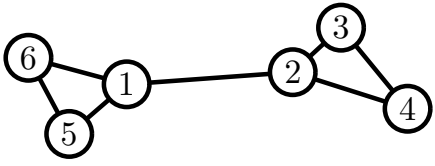}}\quad\subfigure[Data points $y_i(t)$ for $i\in\{1,2,\ldots,6\}$ and $t\in\{0,1,\ldots,6\}$ that form the set of linear equations \eqref{eq:s9}.]{\includegraphics[width=0.48\linewidth]{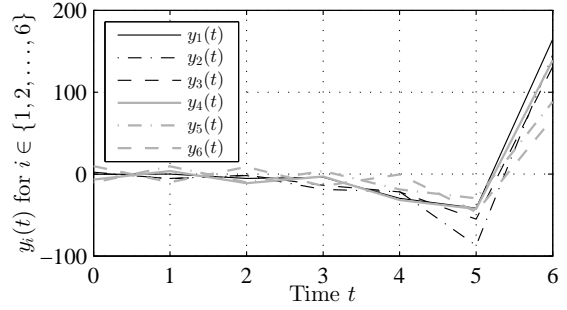}}\\ \subfigure[Node $2$'s estimate $x_2^{(\ell)}(t)$ of the $\ell$th perturbed and true characteristic polynomial coefficients $x^{(\ell)}$ and $\overline{x}^{(\ell)}$ for $\ell\in\{0,1,\ldots,5\}$.]{\includegraphics[width=0.48\linewidth]{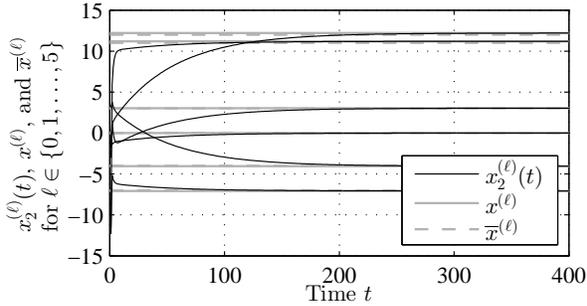}}\quad\subfigure[Node $i$'s estimate $x_i^{(2)}(t)$ of the $2$nd perturbed and true characteristic polynomial coefficients $x^{(2)}$ and $\overline{x}^{(2)}$ for $i\in\{1,2,\ldots,6\}$.]{\includegraphics[width=0.48\linewidth]{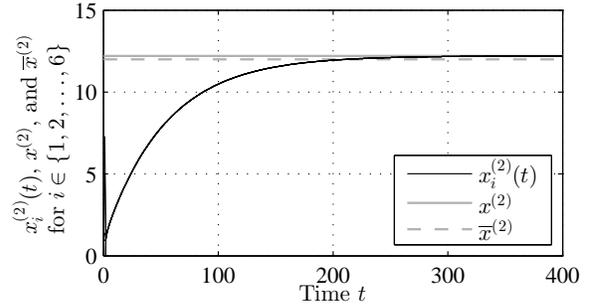}}
\caption{Performance of Algorithm~\ref{alg:notknow} for Scenario~\ref{sce:notknow}.}
\label{fig:sce2}
\end{figure}

Consider next an undirected, connected graph $\mathcal{G}$ with $N=6$ nodes, whose topology is shown in Figure~\ref{fig:sce2}(a). Let $\overline{W}$ represent the adjacency matrix of $\mathcal{G}$ and suppose the nodes wish to determine all the eigenvalues $\overline{\lambda}^{(\ell)}$'s of $\overline{W}$, which are given by $-1.73$, $-1$, $-1$, $-0.41$, $1.73$, and $2.41$. Because they only have local information about $\mathcal{G}$, the nodes do not know whether $\overline{W}$ is cyclic, so that Scenario~\ref{sce:notknow} takes place. (In fact, $\overline{W}$ in this particular example is not cyclic because it is symmetric and has repeated eigenvalues, at $-1$.) Therefore, the nodes have to apply Algorithm~\ref{alg:notknow}. In doing so, they let the perturbation magnitude be $a=0.2$ and obtain from \eqref{eq:wii}--\eqref{eq:wij0} a perturbed matrix $W$ given by
\begin{align*}
W=\begin{bmatrix}0 & 1.04 & 0 & 0 & 1.01 & 0.94\\ 0.98 & 0 & 1.04 & 1.12 & 0 & 0\\ 0 & 0.98 & 0 & 1.06 & 0 & 0\\0 & 0.95 & 1.01 & 0 & 0 & 0\\ 0.98 & 0 & 0 & 0 & 0 & 1.01\\ 0.97 & 0 & 0 & 0 & 0.92 & 0\end{bmatrix},
\end{align*}
whose eigenvalues $\lambda^{(\ell)}$'s are $-1.74$, $-0.97$, $-1.03$, $-0.40$, $1.73$, and $2.43$, which are all distinct and slightly different from the eigenvalues $\overline{\lambda}^{(\ell)}$'s of $\overline{W}$.

Figures~\ref{fig:sce2}(b)--\ref{fig:sce2}(d) display the result of simulating Algorithm~\ref{alg:notknow} with $\alpha_i=100$ $\forall i\in\mathcal{V}$ and $\beta_{\{i,j\}}=10$ $\forall\{i,j\}\in\mathcal{E}$, using a format similar to that of Figures~\ref{fig:sce1}(b)--\ref{fig:sce1}(d). The only difference is that Figures~\ref{fig:sce2}(c) and~\ref{fig:sce2}(d) show not only the characteristic polynomial coefficients $x^{(\ell)}$'s of the ``perturbed'' $W$, but also the characteristic polynomial coefficients $\overline{x}^{(\ell)}$'s of the ``true'' $\overline{W}$. Observe that with Algorithm~\ref{alg:notknow}, the nodes are able to asymptotically determine the $x^{(\ell)}$'s and $\lambda^{(\ell)}$'s. In other words, they are able to approximately calculate the $\overline{x}^{(\ell)}$'s and $\overline{\lambda}^{(\ell)}$'s with small errors.

\section{Conclusion}\label{sec:conc}

In this paper, we have designed and analyzed a two-stage distributed algorithm that enables nodes in a graph to cooperatively estimate the graph spectrum. We have shown that asymptotically accurate estimation can be achieved if the nodes know that the associated matrix is cyclic, and estimation with small errors can be achieved if they do not. As for future research, we believe that making the algorithm numerically more robust, so that it can cope with poorly conditioned $W$ and $A$, is an important next step.

\bibliographystyle{IEEEtran}
\bibliography{paper}

\end{document}